\documentclass[letterpaper,12pt]{article}
\usepackage{verbatim}
\usepackage{xspace}
\usepackage{amsmath,amsfonts,amssymb,amstext,amsthm}
\usepackage[usenames]{color}
\usepackage{hyperref}
\def\bitset{{\{0,1\}}}    
\def\fbitset{{\{-1,1\}}}    
\newcommand{\E}{{\bf E}}
\def\RR{{\mathbb R}}       

\newtheorem{thm}{Theorem}

\newtheorem{lem}[thm]{Lemma}
\newtheorem{fact}[thm]{Fact}
\newtheorem{prop}[thm]{Proposition}
\newtheorem{claim}[thm]{Claim}

\newtheorem*{thm*}{Theorem}
\newtheorem*{lem*}{Lemma}
\newenvironment{thm_again}[1]{\noindent{\bf Theorem #1.~}}{}

\newcommand{\eps}{\epsilon}
\newcommand{\drv}[2]{\frac{\partial #1}{\partial #2}}
\newcommand{\ignore}[1]{}

\begin{document}

\author{Adi Livnat\footnote{Department of Biological Sciences, Virginia Tech.}
\and Christos Papadimitriou\footnote{Computer Science Division, University of California at Berkeley, CA, 94720 USA, christos@berkeley.edu.}
\and Aviad Rubinstein\footnote{Computer Science Division, University of California at Berkeley, CA, 94720 USA, aviad@eecs.berkeley.edu.}
\and Gregory Valiant\footnote{Computer Science Department, Stanford University, CA, 94305.}
\and Andrew Wan\footnote{The Simons Institute for the Theory of Computing, UC Berkeley.  Parts of this work were completed while at Harvard, and at IIIS, Tsinghua University.}
}

\title{Satisfiability and Evolution}

\date{}

\maketitle

\begin{abstract}
\noindent We show that, if truth assignments on $n$ variables reproduce through recombination so that satisfaction of a particular Boolean function confers a small evolutionary advantage, then a polynomially large population over polynomially many generations (polynomial in $n$ and the inverse of the initial satisfaction probability) will end up almost surely consisting exclusively of satisfying truth assignments.  We argue that this theorem sheds light on the problem of the  evolution of complex adaptations.

\end{abstract}
\newpage
\pagenumbering{arabic}
\section{Introduction}\label{sec:intro}
The TCS community has a long history of applying its perspectives and tools to better understand the processes around us; from learning to multi-agent systems, game theory and mechanism design.  By and large, the efforts to understand these areas from a rigorous and algorithmic perspective have been very successful, leading to both rich theories and practical contributions.  
Evolution is, perhaps, one of the most blatantly algorithmic processes, yet our computational understanding of it is still in its infancy (see \cite{valiant2009evolvability} for a pioneering study), and we currently lack a computational theory explaining its apparent success.
Algorithmically, how plausible are the origins of evolution and the emergence of self-replication? Is evolution surprisingly efficient or surprisingly inefficient? What are the necessary criteria for evolution-like algorithms to yield rich, interesting, and diverse ecosystems? Why is recombination (i.e., sexual reproduction) more successful than asexual reproduction? Given the reshuffling of genomes that occurs through recombination, how can complex traits that depend on many different genes arise and spread in a population?  

In this work, we begin to tackle this last question of why complex traits that may depend on many different genes are able to \textit{efficiently} arise in polynomial populations with recombination.  In the standard view of evolution, a variant of a particular gene is more likely to spread across a population if it makes its own contribution to the overall fitness, independent of the contributions of variants of other genes.
How can complex, multi-gene traits spread in a population?  This may seem to be especially problematic for multi-gene traits whose contribution to fitness does not decompose into small additive components associated with each gene variant
---traits with the property that 
even if one gene variant is different from the one that is needed for the right combination,
there is no benefit, or even a net negative benefit.  Here, we provide one rigorous argument for how such complex traits can efficiently spread throughout a population.   While we consider this question in a model that makes considerable (but justifiable) simplifications, this model makes a theoretically rigorous contribution to the fundamental problem of how evolution can produce 
complexity.

\paragraph{Motivating example: Waddington's experiment.}  In 1953 the great experimentalist Conrad Waddington 
exposed the pupae of a population of \textit{Drosophilia melanogaster} to a heat shock, and noticed that in some of the adults that developed, the appearance of the wings had changed (they lacked a complete posterior crossvein) \cite{wad}. He then maintained a population of flies where only those with altered wings were allowed to reproduce. 
  By repeating the procedure of heat shock and selection over the generations, the percentage of flies with altered wings increased over time to values close to one.  Even more interestingly, beginning at generation fourteen, some flies exhibited the new trait even without having been treated with heat shock.  

At first sight, this surprising phenomenon --- known as genetic assimilation --- recalls Lamarck's now discredited belief that acquired traits can be inherited.  However, Boolean functions provide a purely genetic explanation,
which extends the idea originally offered informally by Stern \cite{Stern:1958kx} (see also \cite{Bateman1959b,Falconer1960}):
Suppose that the phenotype ``altered wings'' is a Boolean function of $n$ genes $x_1,\ldots,x_n$ with two alleles (variants) each, thought as $\{-1,1\}$ variables, and of another $\{-1,1\}$ variable $h$ (standing for ``high temperature''). \footnote{Here we assume for simplicity haploid organisms, that is, each individual has only one copy of each gene.}  
$$x_1+x_2+\cdots +x_n + \frac{(1+h)}{2}\cdot k\geq n,$$
for some integer $k$ (think of $n \approx 10$ , and $k\approx n/3$).  

To see how the percentage of flies with altered wings increases in the population over time, we track the allele frequencies from generation to generation.      
Let $\mu_i^t$ be the average value of $x_i$ in the population at time $t$, and assume the genotype frequencies at time $t$ are distributed according to a distribution $\mu^t$ (the reason for denoting the distribution this way will become clear).    If mating occurs at random with free recombination\footnote{See the next section for any unfamiliar terms and concepts from evolution.} then, in expectation, the average value of each $x_i$ in the next generation is given by:  
\begin{equation}\label{eqn:select}
\mu_i^{t+1} = \frac{\E_{\mu^t}[ f(x)\cdot x_i ] }{\E_{\mu^t}[f]},
\end{equation}
where $f(x)=1$ exactly when a fly having genotype $x$ will develop altered wings (i.e., the above inequality is satisfied) and $f(x)=0$ otherwise.  We then assume that the next generation will be distributed according to a product distribution $\mu^{t+1}$, where each $x_i$ has expectation $\mu_i^{t+1}$.  By approximating the genotype frequencies of the population for each generation in this way, 
it can be shown by calculation that a trait with this genotypic specification (a) is very rare in the population under normal temperature $h=1$; (b) it becomes much more common under high temperature $h=1$;  (c) jumps to just above $50\%$ after the first breeding under $h= 1$; (d) after successive breedings with $h=1$ it is nearly fixed; and (e) if after this $h$ becomes $-1$, the trait is still quite common. 

\paragraph{Note: } Our interpretation of Waddington's experiment is a simplification. First, we consider only the distribution of genotypes in each generation to determine the distribution of the next; instead, we could first take a finite sample according to the present distribution and use that sample to calculate the distribution of the next generation.  Such an approximation can only become exact when the population size is infinite, but it is a standard and useful one in population genetics (and we shall eventually consider finite populations for our main result).  We also assumed that
each individual of the new generation is produced by sampling each gene independently of the other genes, and with probability equal to the frequencies of the two alleles of this gene in the parent population (the adults of the previous generation with altered wings).  This assumption turns out to be justified in the settings that we will consider, as will be discussed in the following section.        


\subsection*{Populations of truth assignments}
This way of looking at Waddington's experiment brings about a very natural question:  Is this amplification of satisfying truth assignments (outcomes (c) and (d) of experiment described above) a property of threshold functions, or is it more general?  Does it hold for all monotone functions, for example?  {\em For all Boolean functions?}  

Consider any satisfiable Boolean function $f:\fbitset^n\to \bitset$ of $n$ binary genes (in the absence of the environmental variable $h$ which was crucial in Waddington's phenomenon).  What if genotypes satisfying this Boolean function had a slight advantage 
under natural selection? 
(In Waddington's experiment, they had an absolute advantage because of the experimental design.)  
For example, imagine that genotypes satisfying $f$ survive to adulthood more than the others, in expectation, by a factor of $(1+\epsilon)$, for some small $\epsilon>0$.  Would this trait (that is, satisfaction of the Boolean function $f$) be eventually fixed in the population?  And, if so, 
could this be
a subtle mechanism for introducing complex adaptation in a population?

To reflect our assumption that satisfaction of $f$ confers only an $\epsilon$-advantage, we may take a function  
$f:\fbitset^n \to \{1,1+\epsilon\}$, where we regard the value $1+\epsilon$ as ``satisfied'', and the value $1$ as ``unsatisfied''.  We track the allele frequencies from generation to generation as in Waddington's experiment: Equation \eqref{eqn:select} gives us the average value $\mu_i^{t+1}$ of each $x_i$ in the next generation, and we describe the next generation by the product distribution $\mu^{t+1}$.    

Suppose that we continue this process, starting from distribution $\mu^0$, and defining \\ $\{\mu^1_i\},\{\mu^2_i\},\ldots, \{\mu^t_i\}\ldots$ as above.    Consider the average fitness of the population at time $t$, defined as $\mu^t(f)=\Pr_{\mu^t}[f(x) = 1+\epsilon]$.  The question is, when does $\mu^t(f)$ approach one?  
Our first result states that, for monotone functions, it does after $O\left(\frac{n}{\epsilon \mu^0(f)}\right)$ steps:

\begin{thm}\label{thm:monotone-main}
If $f$ is monotone, then $\mu^t(f) \geq 1-{n (1+\epsilon)\over \epsilon t \mu^0(f)}.$
\end{thm}

\paragraph{Note:}  This nontrivial result also serves to illustrate one point:  The work is {\em not} about satisfiability heuristics (monotone functions are not an impressive benchmark in this regard...).  Heuristics are about finding good {\em individuals} in a population.  In contrast, evolution is about creating good {\em populations}.  This is our focus here.


\medskip 
Our ambition is to prove the same result for all Boolean functions.  Immediately we see that this is impossible if we insist on an infinite population: Consider the function $f=x_1\oplus x_2$: starting with the uniform distribution at time $t=0$, the above dynamics would leave the distribution unchanged, for all time, and hence $\mu^t(f)=1/2$ for all $t$.  The parity function is not the only Boolean function with this property: for example the function ``$\sum_{i=1}^n x_i = k$'', if started at $\mu_i = {k\over n}$, will stay at that spot forever, and will always have $\mu^t(f)=O(1/\sqrt{k})$.  However, experimentation shows that these ``spurious fixpoints'' are not absorbing, and evolution pulls the distribution away from them and towards satisfaction.  That is, this disappointing phenomenon is an artifact of the infinite population simplification. 
Indeed, random genetic drift due to sampling effects has been considered to be a significant component of evolution at the molecular level (it is possible for an allele to become fixed in the population even in the absence of selection). Thus, we need to make the model more realistic.

We adopt a model, consisting of the following process:  At each generation $t$ we create a large population of $N$ individuals (we call this the ``sampling'' step) by 
sampling $N$ times from the product distribution $\mu^t$ to obtain $y^{(1)},\dots,y^{(N)}$ ($N$ is assumed to remain constant from one generation to the next, which is a standard assumption in population genetics \cite{Gillespie2004}).
 The empirical allele frequencies of the sample are given by a vector $\nu^t$, where for each $i$ we have: 
$$\nu^t_i = \frac{1}{N}\sum_{j=1}^{N} y_i^{(j)}.$$       
We write $\nu^t \sim B(\mu^t)$ to denote a draw from this distribution and use $\nu^t$ to denote the implied Bernoulli distribution.  

We then enforce the assumed selection advantage of satisfaction to obtain the ``in-expectation'' frequencies of the subsequent generation:  
\begin{equation}\label{eqn:select-noise}
\mu_i^{t+1} = \frac{\E_{\nu^t}[ f(x)\cdot x_i ] }{\E_{\nu^t}[f]}.
\end{equation}
We show that when selection is weak, \textit{any} satisfiable Boolean function will almost surely be always satisfied after polynomially many time steps.    
\begin{thm} \label{thm:main} {\bf (informal statement)}
For any satisfiable Boolean function $f$ of $n$ variables and any sufficiently small $\epsilon>0$, after $T$ generations of $N$ individuals $\mu^T(f)=1$ with probability arbitrarily close to one, where $T$ and $N$ are polynomial functions of $n$, $1\over \epsilon$, and $1\over \mu^0(f)$.
\end{thm}
The proof of Theorem \ref{thm:main} shows why the population does not become stuck at the previously discussed ``spurious fixpoints;'' 
sampling effects ensure movement over sufficiently many generations, and selection ensures movement is made towards satisfaction.    

\subsection*{Outline of the paper}
In the next section we introduce some basic concepts from population genetics, we define and justify our simplified model, and we present a result due to Nagylaki \cite{Nag} implying that, if selection is weak, then one can assume that the genotype distribution is a product distribution.  In Section \ref{sec:monotone}, we show Theorem \ref{thm:monotone-main} on monotone functions.  Our main result is given in Section \ref{sec:main}, and its proof outlined; the full proof is detailed in the Appendix.  In Section \ref{sec:discussion}, we conclude with a discussion of our result, and a number of open problems.

\section{Evolution background and our model}
The genetic makeup of an organism is its {\em genotype}, which specifies one {\em allele} (gene variant) for each genetic site, or ``locus,'' in the haploid case. We shall be focusing on $n$ specific genes of interest (say, a few dozen out of the many thousands of genes of the species). At each locus, we assume that there are two alleles segregating in the population (hence the relevance of Boolean functions).  Thus, a genotype will be a vector in $\{\pm 1\}^n$. 
We assume the species reproduces sexually (this is crucial, see the discussion in the last section).  In a sexual species reproduction proceeds through {\em recombination,} that is, the formation of a new genotype by choosing alleles from two parental genotypes in the previous generation.  
To produce each generation, the 
individuals
mate at random (we also assume no bipartition into sexes) 
and there is no generation overlap (that is, the new generation is produced en masse just before death of the previous one).  We assume that 
the population size is constant at some large number 
$N$ (expressed as a function of $n$, the number of genes of interest, which is the basic parameter).  Each genotype $g\in\{\pm 1\}^n$ is assumed to have a {\em fitness value} equal to the expected number of offsprings this genotype will produce.  We also assume that the genes {\em recombine freely}, that is, for any two genes $i,j$ of
an offspring,
the probability that the alleles come from the same parent is exactly half (and not larger, as is the case if the two genes are linked).

These assumptions are simplifications of the standard model of population genetics used broadly in the literature, and generally trusted to preserve the essence of selection in sexual populations.  The Boolean assumption is of course meant to bring into play mathematical insights from that field, but we believe that it is not restrictive (for example, allele $-1$ could stand for ``any allele other than allele $1$'').  In this paper we shall make two more assumptions.  The first additional assumption is that the fitness values of our genotypes are either $1$ or $1 + \epsilon$, where $\epsilon>0$ is very small: the organism will reproduce slightly more in expectation if an underlying Boolean function is satisfied.
We discuss this restriction in Section \ref{sec:discussion}.

The final assumption is more problematic in general, but justified in the current context:  We assume that generating an individual 
of the next generation is tantamount to selecting, independently, an allele for each of the $n$ genes, with probability equal to the probability of occurrence of that allele in the parent generation.   That is, we assume that the distribution of the genotypes in a generation is a {\em product distribution}.  This situation is called in the population genetics literature {\em linkage equilibrium}, or {\em the Wright manifold} \cite{Wright:1931uq,Wright:1932kx}.  
In general, genotype frequencies are known to be correlated, and this correlation --- the distance from the product distribution --- is called {\em linkage disequilibrium} \cite{Lewontin:1964ys} and is of importance and interest in the study of evolution.  However, in the absence of selection, a standard argument shows that the distribution of a population quickly reaches linkage equilibrium (arguments exist both for finite and infinite populations).  Our previous assumption places our experiment in a regime known as {\em weak selection}. Weak selection means that the fitness values are in a small interval $[1-\epsilon, 1+\epsilon]$, where $\epsilon$ is called the {\em selection strength}.  An elegant and powerful result due to Thomas Nagylaki \cite{Nag} states that, under weak selection, evolution proceeds to a point very close to linkage equilibrium.  In particular, assume that a population evolves as we described above in a regime of weak selection of strength $\epsilon$, and let $m$ be the total number of  alleles (this is $2n$ in our case; actually, Nagylaki's Theorem also holds under diploid and partial recombination).  By linkage disequilibrium we mean formally  the $L_{\infty}$ distance between the genotype distribution and the product distribution:

\begin{thm} (Nagylaki's Theorem, see \cite{Nag}) \label{thm:nagylaki}  Under weak selection, and after $O(\log m\cdot \log 1/\epsilon)$ generations, linkage disequilibrium is $O(\epsilon)$. 
\end{thm}

In our setting $\epsilon$ is minuscule, so Nagylaki's Theorem motivates our assumption that populations are formed ``by independent sampling of the genetic soup.'' 
We strongly believe that our theorem is true for large $\epsilon$ as well, but this remains open, as discussed in the last section.

\section{Weak selection on monotone functions}
\label{sec:monotone}
In this section we give a self-contained proof of Theorem \ref{thm:monotone-main}.  The proof is simple, once a connection is made to discrete Fourier analysis. In what follows, we assume familiarity with Fourier analysis over the Boolean cube for product distributions.  We briefly review some basic facts and describe the notation used in our proofs.    

For $\mu = (\mu_1,\dots,\mu_n) \in [-1,1]^n$ and a function $f:\fbitset^n_\mu \to \RR$, where 
$\fbitset^n_\mu$ denotes the Boolean cube endowed with the product distribution given by $\mu_i = \E[x_i]$, we 
consider the $\mu$-biased Fourier decomposition of $f$. Let $\sigma_i^2 = 1-\mu_i^2$ be the variance of each bit. 
We denote the $\mu$-biased Fourier coefficients by $\hat{f}(S;\mu) = \E_\mu [f \cdot \phi^\mu_S]$, where $\phi^\mu_S = \prod_{i\in S} \frac{x_i -\mu_i}{\sigma_i}$. 
Let $D_i^{(\mu)} f = \frac{\sigma_i}{2} (f_{i=1} - f_{i=-1})$ be the difference operator for Boolean functions over $\fbitset^n_\mu$.   
We have that 
\begin{equation} \label{eqn:diff} D_i^{(\mu)} f = \sum_{S\ni i} \hat{f}(S;\mu) \phi^\mu_{S\setminus \{i\}},
\end{equation}
and in particular, $\E_\mu[ D_i^{(\mu)}f ] = \hat{f}(i;\mu)$, which we will use repeatedly throughout our proofs. 

Our first step will be to observe that the change in allele frequencies from one generation to the next may be expressed in terms of $f$'s linear Fourier coefficients.   
Let $\mu$ be the vector which specifies the allele frequency of the population at time $t$.  Then, letting $\mu'$ be the allele frequency vector at time $t+1$ and using the selection specified by Equation \eqref{eqn:select}, we have that   
\begin{equation}\label{eqn:select2}
\mu'_i - \mu_i = \sigma_i \frac{\hat{f}(i;\mu)}{\E_\mu[f]}.
\end{equation}
This follows immediately from the definitions:      
\begin{align*}
\sigma_i \cdot \hat{f}(i;\mu) & = \sigma_i\cdot \E_\mu [ f\cdot \phi_i^\mu] \\ 
& = \E_\mu[f\cdot x_i ] - \E_\mu[ f]\cdot \mu_i  \\ 
& = \E_\mu[f]\cdot \mu'_i - \E_\mu[f]\cdot \mu_i. 
\end{align*}

Our proof uses the following well-known facts, which are easily derived from the basic notions 
(see Chapter 2.3, 
\cite{o2014analysis} 
).  First, we have that the influences of a monotone function 
are given by its linear coefficients.  
(For a function $f:\fbitset^n_\mu \to \RR$, we denote its influence in direction $i$ by $\sum_{S\ni i} \hat{f}(S;\mu)^2$.) Next, the inequality of Poincar\'{e} lower bounds the total influence of a function by its variance. The versions below have been scaled to our setting and can be obtained by applying the original facts to a Boolean function $g:\fbitset^n\to \fbitset$ and setting $f(x) = 1+\frac{\epsilon}{2} (1+g)$. 
\begin{prop}\label{prop:influence} Let $f:\fbitset_\mu^n \to \{1,1+\epsilon\}$ be monotone. Then for all $i\in [n]:$
\begin{eqnarray*}
\sum_{S \ni i} \hat{f}(S;\mu)^2  =    \frac{\epsilon \sigma_i}{2} \cdot \hat{f}(i;\mu).
\end{eqnarray*}
\end{prop}

\begin{prop} \label{prop:poinc} Let $f:\fbitset^n_\mu \to \{1,1+\epsilon\}$ and $\operatorname{Var}[f] = 
\E_\mu [f^2] - \E_\mu[f]^2$. Then 
$$\sum_{i \in [n]} \sum_{S\ni i} \hat{f}(S;\mu)^2  = \sum_{S\subseteq [n]} |S|\hat{f}(S;\mu)^2\geq \operatorname{Var}[f].$$
\end{prop}

Equation \eqref{eqn:select2} tells us that the bias of each bit $i$ increases according to the corresponding coefficient $\hat{f}(i)$.  Proposition \ref{prop:influence} tells us that for monotone $f$, the linear coefficients correspond to the influences of $f$.
 Finally, the inequality of Poincar\'{e}
tells us that the linear coefficients must be large.

We may now prove Theorem \ref{thm:monotone-main}.

\begin{thm_again}{\ref{thm:monotone-main}} 
Let $f:\fbitset^n \to \{1,1+\eps\}$ be monotone.  Then $\mu^t(f)\geq 1-\frac{n(1+\epsilon)}{\epsilon t\mu^0(f)}$.  
\end{thm_again}
\begin{proof}
Combining Equation \eqref{eqn:select2} with Propositions \ref{prop:influence} and \ref{prop:poinc} tells us that the sum of the biases increases at each step:
\begin{align*}
\sum_{i\in [n]}(\mu'_i - \mu_i) & =  \frac{2}{\epsilon \cdot \E_\mu[f]} \sum_{i\in [n]}\sum_{S \ni i} \hat{f}(S;\mu)^2 \\
& \geq \frac{2}{\epsilon \cdot \E_\mu[f ] } \operatorname{Var}[f]\\
& = \frac{2}{\epsilon \cdot \E_\mu[f]} \epsilon^2 \mu(f)(1-\mu(f))
\end{align*}
Let $\mu^t(f) = 1- \delta.$   Then for all $t'\leq t$,  the sum of the biases increases at each step: 
\begin{align*}
\sum_{i=1}^n \mu^{t'+1}_i -\sum_{i=1}^n \mu^{t'}_i & \geq \frac{2\eps\mu^{t'}(f)(1-\mu^{t'}(f))}{\E_\mu^{t'}[f]} \\ 
& \geq \frac{2\epsilon \mu^{t'}(f) \delta }{1+\epsilon} \geq \frac{2\delta \epsilon \mu^{0}(f)}{1+\epsilon}.
\end{align*}

On the other hand, we know that $-n \leq \sum_{i=1}^n \mu^{t'}_i \leq n$ for all $t'$, so $t \leq \frac{n (1+\epsilon)}{\delta \epsilon \mu^0(f)}.$
\end{proof}

We remark that Theorem \ref{thm:monotone-main} (with worse parameters) can also be proven using a generalization of the Russo-Margulis lemma to product distributions, which states that the gradient of $\E_\mu[f]$ (as a function of $\mu$) corresponds to the influences of $f$ (see Appendix \ref{sec:prelim}).

\section{The main result} 
\label{sec:main}
  For a function $f:\fbitset^n \to \{1,1+\epsilon\}$, consider the multilinear extension $\tilde{f}:[-1,1]^n\to [1,1+\epsilon]$ defined by $\tilde{f}(\mu) = \E_{x\sim \mu}[f(x)]$.  Our goal is to understand when $\tilde{f}(\mu)=1+\eps$. We start with the precise statement of the main result (compare with Theorem \ref{thm:main}):
\begin{thm}\label{thm:main-full}
Let \textup{$\beta=\sqrt{\frac{\epsilon}{N\left(1-n\epsilon\right)}}$.} If 
\[
\tilde{f}\left(\mu^{0}\right)>1+\sqrt{2 \beta \ln \frac{2}{\beta}}\]
 then there is some constant $C$ such that for any $T\geq C\cdot \frac{\epsilon n^{8}\cdot N^{4}}{1-n\epsilon}$:
 \[
 \Pr\left[\tilde{f}\left(\mu^{\left(T\right)}\right)=1+\epsilon\right]\geq 1-2\beta-2/n.\]

\end{thm}
Note that the conditions in Theorem \ref{thm:main-full} imply restrictions on the initial probability of satisfaction and the strength of selection.  In particular, the selection coefficient must be in the range $1/N^{1/3} < \epsilon < 1/n$ (we discuss this restriction in the next section), and the initial probability of satisfaction must be at least $N^{-1/4}$.  
The full proof of the theorem is given in the Appendix; in this section we sketch its salient points.

One first difficulty in the proof is this:  The convergence proof gauges the improvement in average population fitness obtained during the second of the two steps per generation (the fitness step).  However, the first of the two steps (the sampling step) introduces variance, and we must establish that this variance is insignificant in comparison with the increase in fitness.  Our first lemma (Lemma \ref{lem:noise-fitness}) establishes that the difference between the average fitness of the sample and the average fitness, squared (that is to say, the variance introduced), is bounded from above by the increase in average fitness obtained in the fitness step: 
\begin{align}\label{eqn:variance-lemma}
\E_{\nu\sim B}[ (\tilde{f}(\nu) -\tilde{f}(\mu))^2 ]   \leq 
 \E_{\nu\sim B} [ \tilde{f}(\mu') - \tilde{f}(\nu) ] /   
[(N-1)\cdot (1-n\epsilon)].
\end{align}
Here we focus on one generation, so $\mu$ denotes the product distribution from which the sampling is made, $\nu$ the empirical product distribution of the sample (note that $\tilde{f}(\nu)$ is a random variable with expectation $\tilde{f}(\mu)$), and $\mu'$ the product distribution resulting from the selection (or fitness) step.  Thus, $\mu'$ is the initial product distribution in the next generation. 

To establish inequality \eqref{eqn:variance-lemma}, we first show that the right-hand side is lower bounded by the total mass of the singleton Fourier coefficients of the biased transform (Lemma \ref{lem:density}):
\begin{align}\label{eqn:coord}
\tilde{f}(\mu') - \tilde{f}(\nu) \geq (1-n\epsilon) \sum_{i=1}^n \hat{f}(i;\nu)^2.
\end{align}

The intuition in the proof of \eqref{eqn:coord} is that the fitness step is very close to an $\epsilon$-long step of the gradient ascent of the average fitness function (this intuition is very accurate away from the boundary of the hypercube).  Gradient ascent in each coordinate is captured by the corresponding singleton coefficient squared.  But then there is an analytical complication of approximating the overall ascent by the sum of sequential coordinate-wise ascents; the difficulty is, of course, that the partial derivatives change after each small ascent, and the change must be bounded (Lemma \ref{lem:coord}).  

This establishes that the fitness increase in the selection step is larger than the linear Fourier mass, and hence nonnegative when $\epsilon$ is small.  However, the linear Fourier mass may be zero, as is the case for the exclusive-or function under the uniform distribution (recall the discussion a few lines after Theorem \ref{thm:monotone-main}).  Here, sampling effects will ensure that progress is made in expectation.
We show that, on average, the linear Fourier mass is much larger than the variance (Lemma \ref{lem:variance-noise}):

\begin{align}\label{eqn:lin-coeff}
\E_{\nu\sim B} \left[ (\tilde{f}(\mu) - \tilde{f}(\nu))^2 \right]  \leq 
 \frac{1}{N-1} &\E_{\nu\sim B} \left[ \sum_{i=1}^n \hat{f}(i;\nu)^2 \right]
\end{align}

The rather involved proof of \eqref{eqn:lin-coeff} takes place entirely within the biased Fourier domain (see Appendix \ref{subsec:variance-noise}).  
Now notice that \eqref{eqn:lin-coeff}, combined with \eqref{eqn:coord}, completes the proof of inequality \eqref{eqn:variance-lemma} and Lemma \ref{lem:noise-fitness}.  

Note that the upper bound on the variance in \eqref{eqn:variance-lemma} includes in the denominator a factor of $(1-\epsilon n)\cdot N$.  This immediately tells us that our technique is sharpest when the population $N$ is large and the selection strength $\epsilon$ is small --- in particular, it {\em must} be smaller than $1\over n$.  This latter point is a rather puzzling  limitation of our result:  Why does a theorem about the effectiveness of natural selection become harder to prove when selection is stronger?  One intuitive explanation is that in this case selection works very much like gradient ascent, and it is well known that the convergence of gradient ascent is harder to establish when the ascent step is large, as a large step can ``skip over'' the stationary point sought.  Is this upper limit on $\epsilon$ necessary?  This is an intriguing open question discussed in the last section.

Next, we establish that the total effect of the sampling steps is small:  For any $\alpha > \sqrt{2\beta\ln 2\beta^{-1}}$,
$$\Pr[ | \sum_{t=1}^T \tilde{f}(\nu^t)   -\tilde{f}(\mu^{t-1}) | \geq \alpha ] \leq 2\beta,$$
where $\beta = \left(\frac{\epsilon}{N(1-n\epsilon)}\right)^{1/2}$.

It is not hard to see that the sum  $  \sum_{t=1}^T \tilde{f}(\nu^t)   -\tilde{f}(\mu^{t-1})$ constitutes a martingale, albeit one with no obvious upper bound on each step.   In Lemma \ref{lem:main} we bound the total effect of the sampling step by resorting to a rather exotic martingale inequality derived from a generalization of Bernstein's inequality to martingales with unbounded jumps and proved in \cite{bernstein_generalized} (in fact, a specialization stated in Appendix \ref{sec:martingale} as Lemma \ref{lem:generalized-bernstein}).  

Incidentally, notice that this is the place where it is proved, quite indirectly, that the sampling step succeeds in getting the process unstuck from spurious fixpoints such as $({1\over 2},{1\over 2},\ldots,{1\over 2})$ for the exclusive-or function: 
Since the total effect of sampling is limited, the increase in average fitness must eventually prevail.

Finally, when the process is near a vertex of the hypercube, fitness increases are too small to help finish the argument, but here we rely on the fact that the process is very likely to drift so close to a vertex that it will eventually get stuck there (Lemma \ref{lem:fast-vertex}), completing the proof of the main result.

\section{Discussion}
\label{sec:discussion}

We proved a novel and highly nontrivial aspect of Boolean satisfiability:  By randomly crossing assignments and favoring satisfaction slightly, one can breed a population of pure satisfying truth assignments.  We argued that this rather curious property seems important in understanding one intriguing aspect of evolution: how complex traits controlled by many genes can emerge.

There are many roads of mathematical inquiry opened by this theorem.  First, can the limitations/restrictions of our model be relaxed so that it better reflects the realities of life? 
Some of the assumptions in our model are arguably unrealistic (haploidy, fixed population size, random mating, partly in-expectation fitness calculation), but these follow widely accepted practices in population genetics needed for mathematical simplification. 
  We also make the assumption of weak selection, but this is also a very defensible one for unlinked loci.

There are, however, a few further restrictions of our model that call for discussion: 

\begin{itemize}
\item {\em Two alleles per gene.}  The motivation is, of course, that this assumption ushers in the powerful analytical toolbox of Boolean functions.  We have no doubt that similar results hold for more alleles, but would require a great number of technical adjustments.

\item {\em Fitness landscape.}  We assumed a very specialized fitness landscape with values $1$ and $1+\epsilon$ only.  This is a natural simplification that facilitates the connection to Boolean functions, but we do not believe it is an essential one.  We believe that this result can be extended to much richer landscapes with a small gap, for example to situations in which fitness values are in $[1-\delta, 1]\cup [1+\epsilon, 1+\epsilon +\delta]$ for some small $\delta>0$.  
\end{itemize}

A harder question is, what happens if the fitness gap $\epsilon$ is larger?  As we have mentioned, this is an analytical challenge with roots in the  difficulty of the analysis of gradient descent.  Of course, a constant gap would bring us outside the realm of weak selection, and render our approximation by product distribution baseless.  There are two ways we can proceed:  One is to prove that the exact recurrence equations of genotype frequencies yield eventual satisfaction.  This seems possible but challenging.  

Another avenue, which we have followed for some time, is to work with product distributions anyway.  In particular, what if the fitness landscape has values $\{0,1\}$  --- that is to say, non-satisfying truth assignments are removed from the population, as in Waddington's experiment?  This is a realistic approximation if, for example, this selection does not happen in every generation but every $O(\log n)$ generations (because breeding without selection is known to take you close to the Wright manifold).  
 In such a setting, our quadratic bound for the in-expectation process of monotone functions no longer requires any dependence on the initial probability of satisfaction $\mu^0(f)$.  For the process with sampling, we have the following conjecture.

\paragraph{Conjecture:}  {\em If the fitness landscape has values $\{0,1\},$ then the process reaches near universal satisfaction with probability approaching $1$ as the population size goes to infinity.}

\bigskip  We now want to point out an obvious and yet surprising aspect of our work: 
In the traditional framework of adaptive evolution, each allele spreads in the population mainly either due to an additive contribution to fitness that it makes in and of itself (let us call this ``traditional propagation") or due to random 
genetic drift \cite{:1930fk,Wright:1931uq,Wright:1932kx,WadeGoodnight1998}. 
In our model, however, alleles at different genes are spreading in the population as governed by the complex interactions between them that are continually subject to selection. Thus, a population can change dramatically through a novel process involving subtle changes in genetic statistics and simultaneous gradual emergence in the whole population \cite{Livnat2013,Mayr1963}, and not by traditional propagation.  

Furthermore, notice that since recombination is a crucial ingredient of our analysis, our results inform the question of the role of sex in evolution. In this regard they add to recent works that have begun to examine the role of sex while giving full weight to the importance of genetic interactions \cite{Chastain2014algorithms,Livnat2008mixability}. 

Finally, can our bounds be improved?   
For the monotone case, it is easy to see that the TRIBES function with appropriate fan-in provides a matching lower bound. 
As for the general case, we feel that the very generous bounds of the main result can be improved substantially.  For example, the assumed time bound is only necessary in order to finish the last part of the argument (convergence to a vertex) once the vast majority of the population is already satisfying; more analysis is needed to investigate this subtle phenomenon.

Our proof that the population converges to a single satisfying truth assignment may seem a troubling aspect of our result.  
Two remarks: First, the loss of genetic diversity should not be surprising in itself.  With drift alone, for each locus,  one allele will become fixed eventually (where the probability that a particular allele will be the fixed allele is proportional to its current frequency in the population). 
Second, in our process many satisfying truth assignments are likely to survive for a very long time before the random walk clears the picture.  This fact may be more relevant than the characteristics of eqilibrium; after all, evolution happens in the transient. 


\paragraph{Acknowledgments:}  We are grateful to Yu Liu of Tsinghua University for some very interesting conversations in the beginning of this research.

\bibliography{boolution}
\appendix
\section{Outline of proof}
In the following sections, we prove Theorem \ref{thm:main-full}: 
\begin{thm*}
Let \textup{$\beta=\sqrt{\frac{\epsilon}{N\left(1-n\epsilon\right)}}$.} If 
\[
\tilde{f}\left(\mu^{0}\right)>1+\sqrt{2 \beta \ln \frac{2}{\beta}}\]
 then there is some constant $C$ such that for any $T\geq C\cdot \frac{\epsilon n^{8}\cdot N^{4}}{1-n\epsilon}$:
 \[
 \Pr\left[\tilde{f}\left(\mu^{\left(T\right)}\right)=1+\epsilon\right]\geq 1-2\beta-2/n.\]
\end{thm*}
Our proof of Theorem \ref{thm:main-full} is structured as follows. In Section \ref{sec:variance-noise}, we consider  
the average fitness from one generation to the next. As described in Section \ref{sec:intro}, each generation consists of two steps: the sampling step, which begins with a product distribution $\mu$ and results in an empirical product distribution $\nu$, and the fitness step resulting in a distribution $\mu'$ (which becomes the initial distribution for the next generation). The main lemma (and the most involved step of our proof) of Section \ref{sec:variance-noise} is Lemma \ref{lem:noise-fitness}, which upper bounds the variance of $\tilde{f}(\nu)$, 
by a small fraction of $\E[\tilde{f}(\mu')-\tilde{f}(\nu)]$, the expected increase in average fitness by the fitness step.       

In Section \ref{sec:martingale}, 
we apply Lemma \ref{lem:noise-fitness} with the martingale inequality to prove Lemma \ref{lem:main}, which states that the total fitness decrease will be small with high probability. Finally, we complete the proof of the main theorem in Section \ref{sec:finalproof} by arguing (Lemma \ref{lem:fast-vertex}) that for $T$ as stated in the theorem, $\mu^T$ will reach a vertex of the hypercube (and hence $f(\mu^T) \in \{1,1+\epsilon\}$) with high probability.                 

\section{Selection vs sampling effects}\label{sec:variance-noise}
In this section we consider just one step of the process. 
Let $\mu$ be the initial product distribution of a generation,  
$\nu$ be the empirical product distribution from the sampling step, and $\mu'$ be the product distribution after the fitness step.  
Our main goal in this section is to show that the variance of $\tilde{f}(\nu)$, the average fitness of the population after the sampling step, is small compared to the expected increase in average fitness from the subsequent selection step, $\tilde{f}(\mu')-\tilde{f}(\nu)$.
The main lemma we will prove is the following:    
\begin{lem} \label{lem:noise-fitness} 
Let $\nu$ be the vector of expectations of allele frequencies in the population sample of size $N$, drawn according to $B(\mu)$. Then:
$$\E_{\nu\sim B}[ (\tilde{f}(\nu) -\tilde{f}(\mu))^2 ] \leq
 \E_{\nu\sim B} [ \tilde{f}(\mu') - \tilde{f}(\nu) ]/   
(N-1)\cdot (1-n\epsilon).$$

\end{lem}
We will prove Lemma \ref{lem:noise-fitness} by proving two intermediate lemmas.  First, we show that fitness increase by the selection step $\tilde{f}(\mu') - \tilde{f}(\nu)$ is nearly as large as the $\nu$-biased Fourier weight of the linear coefficients of $f$ (and hence non-negative), provided that  $\epsilon$ is sufficiently small. 
\begin{lem} \label{lem:density} Let $\mu'$ be the expectations of the process after selection from the population $\nu$.   Then:
$$\tilde{f}(\mu') - \tilde{f}(\nu) \geq (1-n\epsilon) \sum_{i=1}^n \hat{f}(i;\nu)^2.$$
\end{lem}
Next, we show that the variance of $\tilde{f}(\nu)$ is at most a small fraction of the expected linear $\nu$-biased Fourier mass of $f$: $\sum_{i=1}^n \hat{f}(i;\nu)^2$---here the expectation is taken over the choice of $\nu.$
\begin{lem}\label{lem:variance-noise}
$$\E_{\nu\sim B} \left[ (\tilde{f}(\mu) - \tilde{f}(\nu))^2 \right] \leq
 \frac{1}{N-1} \E_{\nu\sim B} \left[ \sum_{i=1}^n \hat{f}(i;\nu)^2 \right].$$ 
\end{lem}
Combining Lemmas \ref{lem:density} and \ref{lem:variance-noise} gives us Lemma \ref{lem:noise-fitness}.  

\subsection{Preliminaries} \label{sec:prelim}
Throughout our proofs, we will use the notation and basic facts established at the beginning of Section \ref{sec:monotone}.  At times, we will use biased Fourier analysis with different product distributions $\mu$ and $\mu'$ at the same time.  To prevent ambiguity, we will refer to the standard deviation of the $i$'th bit as $\sigma_i(\mu) = 1-\mu_i^2$ (similarly for $\sigma_i(\mu'))$, but we will use $\sigma_i$ when the context makes the distribution clear.           

Recall that the exentsion of $f:\fbitset^n\to \{1,1+\epsilon\}$,  
$$\tilde{f}(\mu) = \E_{x\sim \mu}[ f(x) ]=\sum_{S\subseteq [n]} \hat{f}(S) \prod_{j\in S} \mu_j $$
is multilinear. Note that its derivative in the $i$'th direction is given by 
\begin{align}
 \drv{\tilde{f}(\mu)}{\mu_i} &=  \sum_{S\ni i} \hat{f}(S) \prod_{j\in S\setminus i} \mu_j =\E_\mu[ D^{(1/2)}_i f] \nonumber \\ 
 &=  \frac{1}{\sigma_i}\E_\mu[ D^{(\mu)}_i f]  \nonumber \\
 &= \frac{\hat{f}(i;\mu)}{\sigma_i}.
\label{eqn:russo}
 \end{align}
 Here (\ref{eqn:russo}) is a straightforward generalization of the Russo-Margulis Lemma~\cite{Margulis:74, Russo:78} for product distributions. Thus, we may write the change in allele frequency from the fitness step as: 
\begin{equation}\label{eqn:prob-step}
\mu'_i - \nu_i =   \sigma_i(\nu) \frac{\hat{f}(i;\nu)}{\E_{\nu}[f]}= \frac{\sigma_i^2(\nu)}{\tilde{f}(\nu)} \drv{\tilde{f}(\nu)}{\nu_i},
\end{equation}
where the first equality holds by Equation \eqref{eqn:select2} (derived in Section \ref{sec:monotone}).

\subsection{Proof of Lemma \ref{lem:density}}
We will compute the fitness change as each coordinate changes.   
Consider the hybrid distributions given by the expectations 
$w^i = (\mu'_1,\ldots,\mu'_{i}, \nu_{i+1},\ldots,\nu_n)$
so that 
$w^0 = \nu$ and $w^n = \mu'$. We have that 
$$ \tilde{f}(\mu') - \tilde{f}(\nu) = \sum_{i=1}^n  \tilde{f}(w^{i}) - \tilde{f}(w^{i-1}).$$
Observe that the first increment is easily computed as
$$\tilde{f}(w^1)- \tilde{f}(\nu) = (\mu'_1-\nu) \drv{\tilde{f}(\nu)}{\nu_1} = \frac{\hat{f}(1;\nu)^2}{\E_\nu[f]}.$$ 
However, this formula will not be valid for subsequent hybrids as the derivatives of $\tilde{f}$ have changed. 
We start by showing that the derivatives in each direction do not change much between the hybrid distributions. The lemma below will allow us to approximate the derivatives of the hybrids by the derivative of $\tilde{f}$ at $\nu$.
\begin{lem}\label{lem:coord}
Let $\nu'\geq \nu \in [+1,-1]^n$  differ only on the $j$-th coordinate, i.e., $\nu_\ell=\nu'_\ell$ for all $\ell\neq j$ and $\nu'_{j} > \nu_{j}$.  Then for any $i\geq 2$ and $j<i$, 
$$\drv{\tilde{f}(\nu')}{\nu'_i}  - \drv{\tilde{f}(\nu)}{\nu_{i}} = (\nu'_j-\nu_j) \frac{\hat{f}(\{i,j\};\nu)}{\sigma_i \sigma_j}.$$
In particular, for the hybrid distributions above with $i\geq 2$: 
$$
\drv{\tilde{f}(w^{i-1})}{w^{i-1}_i} - \drv{\tilde{f}(\nu)}{\nu_i}  = 
\frac{1}{\E_\nu[f]\sigma_i} \sum_{j=1}^{i-1} \hat{f}(j;\nu) \hat{f}(\{i,j\};\nu).$$
\end{lem}

\begin{proof}
Expanding the derivatives in terms of the unbiased Fourier coefficients, we have that 
\begin{align*}
\drv{\tilde{f}(\nu')}{\nu'_i} - \drv{\tilde{f}(\nu)}{\nu_i} & =   
\sum_{S\ni i} \hat{f}(S) \left( \prod_{\ell \in S\setminus i} \nu'_\ell -\prod_{\ell \in S\setminus i} \nu_\ell \right) \\
& =  \sum_{S\ni \{j,i\}} \hat{f}(S) (\nu'_j- \nu_j) \prod_{\ell \in S\setminus\{j,i\}} \nu_\ell \\ 
& =  (\nu'_j -\nu_j) E_\nu [ D^{(1/2)}_i D^{(1/2)}_j f].\\
\end{align*}
The proof of the first equality in the lemma is completed by noting that:
$$\E_\nu[ D^{(1/2)}_i D^{(1/2)}_j f]  =  \frac{1}{\sigma_i\sigma_j}\E_\nu [ D_i^{(\nu)}D_j^{(\nu)} f].$$
For the second equality of the lemma, 
we may write the difference between the derivative in the $i$'-th direction under the hybrid distribution $w^{i-1}$ and under the original distribution as a telescoping sum: 
\begin{align*}
\drv{\tilde{f}(w^{i-1})}{w^{i-1}_i} - \drv{\tilde{f}(\nu)}{\nu_i}  & = 
\sum_{j=1}^{i-1} (w^{j}_j-w^{j-1}_j) \frac{\hat{f}(\{i,j\});\nu)}{\sigma_i\sigma_j} \\ 
& =  \sum_{j=1}^{i-1} \frac{\sigma_j \hat{f}(j;\nu)}{\E_\nu[f]} \frac{\hat{f}(\{i,j\};\nu)}{\sigma_i\sigma_j} \\
& =  \frac{1}{\E_\nu[f]\sigma_i} \sum_{j=1}^{i-1} \hat{f}(j;\nu)\hat{f}(\{i,j\};\nu).
\qedhere
\end{align*}
\end{proof}
 We are now ready to prove Lemma \ref{lem:density}: 
\begin{lem*}[\ref{lem:density}]
Let $\mu'$ be the expectations of the process after selection from the population $\nu$.   Then:
$$\tilde{f}(\mu') - \tilde{f}(\nu) \geq (1-n\epsilon) \sum_{i=1}^n \hat{f}(i;\nu)^2.$$
\end{lem*}

\begin{proof}
We will bound each of the differences between the hybrid densities in the summation:
$$ \tilde{f}(\mu') - \tilde{f}(\nu) = \sum_{i=1}^n  \tilde{f}(w^{i}) - \tilde{f}(w^{i-1}) .$$
Since $\tilde{f}$ is multilinear,  we have that for $i\geq 2:$ 
\begin{align*}
 \tilde{f}(w^{i})-\tilde{f}(w^{i-1}) & =  (\mu'_i - \nu_i )  \drv{\tilde{f}(w^{i-1})}{w_i^{i-1}}\\
 & =  \frac{\sigma_i\hat{f}(i;\nu)}{\E_\nu[f]} 
 \left(\drv{\tilde{f}(\nu)}{\nu_i}+\drv{\tilde{f}(w^{i-1})}{w_i^{i-1}} - \drv{\tilde{f}(\nu)}{\nu_i}\right) \\
 & =  \frac{\sigma_i\hat{f}(i;\nu)}{\E_\nu[f]} 
 \left(\frac{\hat{f}(i;\nu)}{\sigma_i}+ 
 	\frac{1}{\E_\nu[f]\sigma_i}\sum_{j=1}^{i-1} \hat{f}(j;\nu)\hat{f}(\{i,j\};\nu) \right).
\end{align*}
Recalling that $\hat{f}(\{i,j\};\nu) = \sigma_i \sigma_j \E_\nu[ D^{(1/2)}_i D^{(1/2)}_j f]$ and using the fact that $| D^{(1/2)}_i D^{(1/2)}_j f(x)| \leq \epsilon/4$
for any $x$ and $\E_\nu[f]\geq 1$, we have:  
\begin{equation} \label{eqn:err} 
\frac{1}{\E_\nu[f]\sigma_i}\left|  \sum_{j=1}^{i-1} \hat{f}(j;\nu) \hat{f}(\{i,j\},\nu)  \right| \leq 
\epsilon \sum_{j=1}^{i-1} \sigma_j |\hat{f}(j;\nu)|. \end{equation}
Assume WLOG that $\sigma_i |\hat{f}(i;\nu)| \geq \sigma_j |\hat{f}(j;\nu)|$ for all $j<i$. so that \eqref{eqn:err} is at most 
$\frac{\epsilon}{4} (i-1) \sigma_i |\hat{f}(i;\nu)|$.
Substituting into the telescoping sum, we have: 
\begin{align*}
\sum_{i=1}^n \tilde{f}(w^i) - \tilde{f}(w^{i-1}) & \geq \frac{1}{\E_\nu[f]} \sum_{i=1}^n \sigma_i\hat{f}(i;\nu) \left(\frac{\hat{f}(i;\nu)}{\sigma_i} - \frac{\epsilon}{4} (i-1) \sigma_i |\hat{f}(i;\nu)|\right)  \\ 
& \geq \frac{1}{1+\epsilon}\sum_{i=1}^n \hat{f}(i;\nu)^2 (1-\frac{\epsilon}{4}(i-1)\sigma_i^2 )   \\ 
& \geq  (1-n\epsilon) \sum_{i=1}^n \hat{f}(i;\nu)^2.
\qedhere
\end{align*}
\end{proof}

\subsection{Proof of Lemma \ref{lem:variance-noise}}\label{subsec:variance-noise}
Our goal is to show that: 
$$ \E_{\nu\sim B} \left[ \sum_{i=1}^n \hat{f}(i;\nu)^2 \right]  \geq (N-1) \cdot \E_{\nu\sim B} \left[ (\tilde{f}(\mu) - \tilde{f}(\nu))^2 \right].$$ 
The first key observation is that the Fourier basis with respect to the product distribution $\mu$ is still orthogonal with respect to $B$, i.e., 
we have $\E_{\nu\sim B} [ \phi^\mu_S(\nu) \phi^\mu_T(\nu)]=0$ for $S\neq T$, 
because $B$ is a product distribution and $\E_{\nu \sim B} [ \phi^\mu_S (\nu)]=0$ for $S\neq \emptyset$. 
In particular, Parseval's holds for the extension of $f$ with respect to $B$: 
\begin{claim}\label{eqn:parsevals} 
$$\E_{\nu\sim B}[\tilde{f}(\nu)^2] =  \E_{B} \left[ \left(  \sum_{S} \hat{f}(S;\mu) \phi^\mu_{S} (\nu) \right)^2 \right] = \sum_{S} \hat{f}(S;\mu)^2 \E_{\nu\sim B}[ \phi_{S} ^\mu(\nu)^2 ].$$
\end{claim}
Our approach will be to consider both sides of the inequality using the $\mu$-biased Fourier basis of $f$.  This is straightforward for the variance of $\tilde{f}(\nu)$ using Parseval's.  For the right hand side,
we observe that the $\nu$-biased linear coefficients may be viewed as functions in $\nu$.  In fact, each linear coefficient can be viewed as the extension of the $\mu$-biased difference operator to $[-1,1]^n$, modulo  
a normalization factor.  
\begin{claim} \label{eqn:lin-mu}
$$\hat{f}(i;\nu) = \frac{\sigma_i(\nu)}{\sigma_i(\mu)} \sum_{S\ni i} \hat{f}(S;\mu)\phi^\mu_{S\setminus i}(\nu).$$
\end{claim}
\begin{proof}
We rewrite the linear $\nu$-biased Fourier coefficient in terms of the $\mu$-biased difference operator: 
\begin{align*}
 \hat{f}(i;\nu) = \E_{x\sim \nu} [ D_i^{(\nu)} f  ]  &=  \sigma_i(\nu) \E_\nu [  D_i^{(1/2)} f] \\
 &= \frac{\sigma_i(\nu)}{\sigma_i(\mu)}\E_\nu \left[  D_i^{(\mu)}f \right] \\ 
 &= \frac{\sigma_i(\nu)}{\sigma_i(\mu)} \sum_{S\ni i} \hat{f}(S;\mu)   \E_\nu[ \phi^\mu_{S\setminus i}]\\
 &= \frac{\sigma_i(\nu)}{\sigma_i(\mu)} \sum_{S\ni i} \hat{f}(S;\mu)\phi^\mu_{S\setminus i}(\nu), 
\end{align*}
where the penultimate equality holds by the definition of $D^{(\mu)}_i$, and the final holds because $\nu$ is a product distribution. 
\end{proof}

Finally, we will use the fact that the variance of $\phi^\mu_i(\nu)$ grows smaller as the sample size increases:    
\begin{fact}
$$\E_{\nu\sim B}[  \phi^\mu_{S\setminus i}(\nu)^2 ]  = \E_{\nu\sim B}[ \phi^\mu_{S}(\nu)^2 ]\cdot N.$$
\label{eqn:lin-var}
\end{fact}
\begin{proof}
Because $B$ is a product distribution, we have:
$$\E_{\nu\sim B} [ \phi_{S}^{\mu}(\nu)^2 ] = \E_{\nu\sim B}[\phi_{S\setminus i}^{\mu}(\nu)^2]\cdot \E_{\nu\sim B}[\phi_i^\mu(\nu)^2].$$
Then 
\begin{align*}
\E_{\nu\sim B}[ \phi_i^\mu(\nu)^2 ] & = \E_{\nu\sim B} [ (\nu_i -\mu_i)^2/\sigma^2_i(\mu)]\\
& = \frac{ \E_{\nu\sim B}[(\nu_i-\mu_i)^2]}{1-\mu_i^2}\\
& = \frac{\E_B[ \nu_i^2] -\mu_i^2}{1-\mu_i^2}\\
& = \frac{N^{-1} +\mu_i^2- N^{-1}\mu_i^2 -\mu_i^2}{1-\mu_i^2}.
\qedhere
\end{align*}
\end{proof}

\medskip 

\noindent With the previous claims 
in hand, we are ready to prove Lemma \ref{lem:variance-noise}: 
\begin{lem*}[\ref{lem:variance-noise}]
$$\E_{\nu\sim B} \left[ (\tilde{f}(\mu) - \tilde{f}(\nu))^2 \right] \leq
 \frac{1}{N-1} \E_{\nu\sim B} \left[ \sum_{i=1}^n \hat{f}(i;\nu)^2 \right].$$ 
\end{lem*}

\begin{proof}
We first consider the expected $\nu$-biased linear Fourier weight. Applying Claims \ref{eqn:lin-mu} and \ref{eqn:parsevals}, and summing over $i$ we have: 
\begin{align*}
\sum_{i=1}^n E_{\nu\sim B}[ \hat{f}(i;\nu)^2 ] &= \sum_{i=1}^n \frac{\E_B[ \sigma_i(\nu)^2]}{\sigma_i(\mu)^2}\sum_{S\ni i} \hat{f}(S;\mu)^2 \E_B[\phi_{S\setminus i}^\mu(\nu)^2] \\
&= \left(1-\frac{1}{N}\right)\sum_{i=1}^n \sum_{S\ni i} \hat{f}(S;\mu)^2 \E_B[ \phi_{S\setminus i}^\mu(\nu)^2 ]  \\
&= \left(1-\frac{1}{N}\right)\sum_{i=1}^n \sum_{S\ni i} \hat{f}(S;\mu)^2 \E_B[ \phi_S^\mu(\nu)^2 ] N \\
&\geq (N-1)\cdot \sum_{S\neq \emptyset} \hat{f}(S;\mu)^2 \E[ \phi_S^\mu(\nu)^2] .
\end{align*}
Note in the first equality that the $\sigma_i(\nu)^2$ depends only on $\nu_i$, while $\phi_{S\setminus i}^\mu$ does not, so the expectations may be taken separately.   
For the next equality, we calculate  
\begin{align*}
\E_B[\sigma_i(\nu)^2] &= 1-E_B[\nu_i^2] \\ 
 & = 1 - \left(\mu_i^2+ \frac{1-\mu_i^2}{N} \right),
\end{align*}
which gives that 
$$\frac{\E_B[ \sigma_i(\nu)^2 ] }{ \sigma_i(\mu)^2} = 1-1/N.$$
The third equality holds by Fact \ref{eqn:lin-var}, and the final inequality holds since each non-empty coefficient appears in the sum at least once.

Using Claim \ref{eqn:parsevals} and rewriting the variance of $\tilde{f}(\nu)$ using the $\mu$-biased Fourier basis for $f$, we have:
\begin{align*}
\E_{\nu\sim B} [ (\tilde{f}(\nu) - \tilde{f}(\mu))^2 ] &= \E_B [ (\sum_{S\neq \emptyset} \hat{f}(S;\mu) \phi_S^\mu(\nu))^2 ] \\
&= \sum_{S\neq \emptyset} \hat{f}(S;\mu)^2 \E_B[ \phi_S^\mu(\nu)^2].  \qedhere
\end{align*}
\end{proof}

\section{Bounding the cumulative effect of sampling} \label{sec:martingale}
We saw in Lemma \ref{lem:density} that the selection step always results in a non-negative change in fitness when $\epsilon$ is sufficiently small.  The sampling steps, however, may decrease fitness.  In this section we show that the cumulative effect of sampling on fitness will be small.  We use $\mu^{t}$ to denote the initial distribution of the generation at time $t$, and $\nu^{t}$ to denote the product distribution after the sampling step. according to $B(\mu^{t})$.  Then the selection step determines the population at time $t+1$ which we write as $\mu^{t+1}$ (determined by $\nu^{t}$).        
By Lemma \ref{lem:noise-fitness} at each stage the variance of $\tilde{f}(\nu)$ is a small fraction of the expected fitness increase after selection.  Summing over all generations, the total variance from the sampling steps is a small fraction of the total fitness increase from the selection steps.  Finally, we bound from above the total fitness decrease due to sampling effects; for this last step we need the following generalization of Bernstein's inequality to martingales with unbounded jumps by Dzhaparidze and Zanten:  
\begin{lem} 
\label{lem:generalized-bernstein}(Theorem 3.3%
\footnote{This is a special case of their Theorem 3.3, which corresponds, in
their notation, to the limit as $a\rightarrow0$.%
}, \cite{bernstein_generalized}) Let $\left\{ {\cal F}_{t}\right\} _{t=0,1,\dots}$
be a filtration, and let $\zeta_{1},\zeta_{2},\dots$ be a martingale
difference sequence w.r.t. $\left\{ {\cal F}_{t}\right\} $. Consider
the martingale $S_{T}=\sum_{t=1}^{T}\zeta_{t}$. Define: 
\begin{align*}
H_{T} = \sum\zeta_t^2 +\sum\mbox{E}\left[\zeta_{t}^{2}\mid F_{t-1}\right]
\end{align*}
Then, for each stopping time $\tau$,\[
\Pr\left[\max_{T\leq\tau}\left|S_{T}\right|>z,H_{\tau}\leq L\right]\leq2\exp\left(\frac{-z^{2}}{2L}\right)\]
\end{lem}

Using this machinery, we can show that the total decrease of fitness due to sampling is small:

\begin{lem}\label{lem:main}
Let $\beta = \left(\frac{2\epsilon}{(N-1)(1-n\epsilon)}\right)^{1/2}$ and $\alpha = \sqrt{2\beta \ln \frac{2}{\beta}}$. 
Then 
$$\Pr[ | \sum_{t=1}^T \tilde{f}(\nu^t)   -\tilde{f}(\mu^{t-1}) | \geq \alpha ] \leq 2\beta.$$
\end{lem}
\begin{proof}
For each sequence $\left(\mu^{t}\right)_{t=0}^{T}$ of
populations up to time $T$, define its congruence class as a subset
of infinite sequences:\[
\left[\left(\mu^{t}\right)_{t=0}^{T}\right]=\left\{ \left(w^{t}\right)_{t=0}^{\infty}\colon w^{t}=\mu^{t}\forall t\leq T\right\} \]
Now, for a time $T$, consider the space of possible sequences of populations: 
\[ \mathcal{F}_{T}=\left\{ \left[\left(\mu^{t}\right)_{t=0}^{T}\right]\right\}.\]
Then  $\mathcal{F}_{0}\subset\mathcal{F}_{1}\subset\dots$ is a filtration.  We will consider the following martingale:\[
S_{T}= \sum_{t=0}^T\zeta_t  = \sum_{t=0}^{T}\tilde{f}\left(\nu^{t}\right)-\tilde{f}\left(\mu^{t}\right)\]
Notice that this is indeed a martingale because\begin{eqnarray*}
\mbox{E}\left[S_{T}\mid\mathcal{F}_{T-1}\right] & = & \mbox{E}\left[\sum_{t=0}^{T}\tilde{f}\left(\nu^{t}\right)-\tilde{f}\left(\mu^{t}\right)\mid\mathcal{F}_{T-1}\right]\\
 & = & S_{T-1}+\mbox{E}\left[\tilde{f}\left(\nu^{T}\right)-\tilde{f}\left(\mu^{T}\right)\mid\mathcal{F}_{T-1}\right]\\
 & = & S_{T-1}.\end{eqnarray*}
To apply Lemma \ref{lem:generalized-bernstein},
we also define the following sequences:\begin{eqnarray*}
M_{T} & = & \sum_{t=0}^{T}\left(\tilde{f}\left(\nu^{t}\right)-\tilde{f}\left(\mu^{t}\right)\right)^{2}\\
V_{T} & = & \sum_{t=0}^{T}\mbox{E}\left[\left(\tilde{f}\left(\nu^{t}\right)-\tilde{f}\left(\mu^{t}\right)\right)^{2}\mid\mathcal{F}_{t-1}\right]\\
H_{T} & = & M_{T}+V_{T}\end{eqnarray*}
For each $T$, we show that $\Pr[H_T \geq \beta] \leq \beta$ by bounding $\E[ H_T]$ and applying Markov's inequality.  Applying Lemma \ref{lem:noise-fitness}, we have that \begin{align*}
\E[ M_T ] &= \E \left[ \sum_{t=0}^T (\tilde{f}(\nu^t) - \tilde{f}(\mu^{t}))^2 \right]  \\ 
& \leq \frac{1}{(N-1)(1-n\epsilon)}\cdot \E \left[ \sum_{t=0}^T \tilde{f}(\mu^{t+1}) - \tilde{f}(\nu^t) \right]  \\ 
& \leq \frac{\epsilon}{(N-1)(1-n\epsilon)},
\end{align*}
where the last inequality holds because
$$
 \E\left[ \sum_{t=0}^T \tilde{f}(\mu^{t+1}) - \tilde{f}(\nu^t) +  \sum_{t=0}^T \tilde{f}(\nu^t) - \tilde{f}(\mu^{t}) \right] = 
\E[ \tilde{f}(\mu^{T+1}) -\tilde{f}(\mu^0)] \leq \epsilon, 
$$
and $\E[S_T]= \E[ \sum_{t=0}^T \tilde{f}(\nu^{t})-\tilde{f}(\mu^{t})]=0.$ 
Similarly we may apply Lemma \ref{lem:noise-fitness} to each term of $V_T$: \begin{align*}
 \E \left[ (\tilde{f}(\nu^t) - \tilde{f}(\mu^{t}))^2 | \mathcal{F}_{t-1} \right] 
&= \E_{\nu^t \sim B(\mu^{t})} \left[ (\tilde{f}(\nu^t) - \tilde{f}(\mu^{t}))^2 \right] \\
& \leq \frac{1}{(N-1)(1-n\epsilon)} \cdot \E_{\nu^t \sim B(\mu^{t})} \left[ \tilde{f}(\mu^{t+1}) - \tilde{f}(\nu^{t}) \right].\\
\end{align*}
Summing over $t$ and taking the expectation, we again have that $\E[V_T] \leq \epsilon/(N-1)(1-n\epsilon)$. Thus, we have that $\E[H_T] = \E[M_T+V_T] \leq \frac{2\epsilon}{(N-1)(1-n\epsilon)} \leq \beta^2$.  
Finally, applying Lemma \ref{lem:generalized-bernstein} to $S_T$, we have 
$$\Pr[ \max_{T \leq \tau} |S_T| \geq \alpha, H_T \leq \beta] \leq 2 \exp \left(-\frac{\alpha^2}{2\beta}\right)
\leq \beta.$$ 
Combining this with the bound $\Pr[H_T \geq \beta]\leq \beta$ gives the lemma.   
\end{proof}

\section{Proof of the main theorem} \label{sec:finalproof}
To complete the proof of Theorem \ref{thm:main-full},
we first show (Lemma \ref{lem:fast-vertex} below) that for sufficiently large $T$, the population $\mu^{T}$ is at a vertex of the Boolean cube with high probability.  Finally, we combine this with Lemma \ref{lem:main}, which bounds the probability that $\tilde{f}(\mu^T)\neq 1+\epsilon$. 
\begin{lem} \label{lem:fast-vertex}
There is a constant $C>0$ such that 
for any $T\geq C\cdot \frac{\epsilon n^8 N^4}{1-n\epsilon}$,  we have:
$$ \Pr[ \mu^{T} \notin \fbitset^n ] < 2/n.$$
\end{lem}
\begin{proof}
Note that if $|\nu_j^{t'}| = 1$ for some time $t'$, then $\nu_j^{t} = \nu_j^{t'}$ for every $t\geq t'$.  Observe also that if $|\mu_j^{t'}| > 1-(n^2N)^{-1}$ (in this case we say $j$ is $\alpha$-determined with $\alpha= n^{-2}N^{-1}$), we have by Markov's inequality: 
$$\Pr[ |\nu_j^{t'}| < 1 ] \leq 1/n^2.$$
We will show that after enough time, it is unlikely that there is any coordinate that was never $\alpha$-determined.  More precisely, 
let $A_{j;t}$ be the event that coordinate $j$ was not $\alpha$-determined for $\mu^{1},\dots,\mu^{t}.$
To prove the lemma, the above reasoning tells us that it suffices to show that for $T$ as set in the condition of the lemma: 
$$ \Pr[ \bigvee_{j=1}^n A_{j;T}] \leq 1/n. $$
We will consider each coordinate separately and show that for each $j$, 
$\Pr[ A_{j;T} ] \leq n^{-2}.$
Our proof will use the following simple claims relating $\Pr[A_{j;T}]$ to the selection steps of the process.   
\begin{claim}\label{claim:determined-variance}
Fix any time $t_0$ and an interval $T_1,$ such that $t_0+T_1\leq T$.  Then: 
$$ \E\left[ \left(\sum_{t=t_0}^{t_0+T_1}  \nu_j^{t}- \mu_j^{t}\right)^2 \right] \geq \frac{\alpha\Pr\left[A_{j;T}\right]T_{1}}{2N}.
$$
\end{claim}
\begin{proof}
\begin{align*}
 \E\left[ \left(\sum_{t=t_0}^{t_0+T_1} \nu_j^{t}- \mu_j^{t}\right)^2 \right]  & = 
 \sum_{t=t_0}^{t_0+T_1} \E_{\nu^{t} \sim B(\mu^{t})} \left[ \left(\nu_j^{t} -\mu_j^{t}\right)^2\right] \\
  & \geq  \sum_{t=t_0}^{t_0+T_{1}}\Pr\left[A_{j;t}\right]\E_{\nu^{t}}\left[\left(\nu_{j}^{t}-\mu_{j}^{t}\right)^{2}\mid A_{j;t}\right]\nonumber \\
 & \geq  \frac{\alpha\Pr\left[A_{j;T}\right]T_{1}}{2N}.
\end{align*}
Note that for $t < t'$, the outcome of $(\nu_j^{(t')} - \mu_j^{(t')})$ has expctation $0$, even given any information about time $t$; this gives the first equality.  
The last inequality holds because $\Pr[A_{j;t}] \leq \Pr[ A_{j;t'}]$ for $t\geq t'$ and because
$$ \E_{\nu^{t}\sim B\left(\mu^{t}\right)}\left[\left(\nu_{j}^{t}-\mu_{j}^{t}\right)^{2}\mid A_{j;t}\right] = \frac{\sigma_j(\mu^{t})^2 }{N} \geq \frac{\alpha}{2N}. \qedhere
$$
\end{proof}

The next claim tells us that for any interval of time $t_0\dots(t_0+T_1),$ the change in $\mu_j$ due to the sampling steps cannot be much more than the change from the selection steps.  
\begin{claim}\label{claim:noise-select}
$$ \frac{1}{2}\left(\sum_{t=t_0}^{t_0+T_1} \nu_j^{t} - \mu_j^{t}\right)^2 \leq 
\left(\sum_{t=t_0}^{t_0+T_1} \mu_j^{t+1}-\nu_j^{t}\right)^2+4.$$
\end{claim}
\begin{proof}
Observe that
 $$\mu_j^{t_0+T_1+1} - \mu_j^{t_0} = \left(\sum_{t=t_0}^{t_0+T_1} \nu_j^{t} - \mu_j^{t}\right) +\left(\sum_{t=t_0}^{t_0+T_1}\mu_j^{t+1}-\nu_j^{t}\right)$$
has magnitude at most 2, which gives:  
$$ \left| \sum_{t=t_0}^{t_0+T_1} \nu_j^{t} - \mu_j^{t} \right| \leq \left|\sum_{t=t_0}^{t_0+T_1}\mu_j^{t+1}-\nu_j^{t} \right|+2.$$ 
Squaring both sides, the claim follows from the fact that 
$2|x|^2+8 \geq (|x|+2)^2 . \qedhere$
\end{proof}

We now complete the proof of the lemma.  First, combining Claims \ref{claim:determined-variance} and \ref{claim:noise-select} tells us that for $t_0+T_1\leq T$: 
$$ \E\left[\left(\sum_{t=t_0}^{t_0+T_1} \mu_j^{t+1}-\nu_j^{t}\right)^2\right] \geq 
\frac{\alpha\Pr\left[A_{j;T}\right]T_{1}}{4N} - 4. 
$$
By applying Cauchy-Schwarz and Lemma \ref{lem:density}, we can relate the quantity inside the expectation to the change in fitness:
\begin{align*}
 \left(\sum_{t=t_0}^{t_0+T_1} \mu_j^{t+1}-\nu_j^{t}\right)^2  
& \leq T_1 \sum_{t=t_0}^{t_0+T_1} \left(\mu_j^{t+1}-\nu_j^{t}\right)^2 \\
& = T_1 \sum_{t=t_0}^{t_0+T_1} \left(\frac{\sigma_i(\nu^t)\hat{f}(j;\nu^t)}{\tilde{f}(\nu^t)}\right)^2\\
& \leq \frac{T_1}{1-n\epsilon}\sum_{t=t_0}^{t_0+T_1} \tilde{f}(\mu^{t+1})-\tilde{f}(\nu^t)
\end{align*}
Taking expectations on both sides, we conclude that for a sufficiently long interval, the expected change in fitness on that interval is a good upper bound on $\frac{\alpha \Pr[ A_{j;T}]}{4N}$: 
\begin{equation}\label{eqn:interval}
\frac{1}{1-n\epsilon}\E\left[ \sum_{t=t_0}^{t_0+T_1} \tilde{f}(\mu^{t+1})-\tilde{f}(\nu^t) \right] \geq
\frac{\alpha\Pr[A_{j;T}]}{4N}-4/T_1.
\end{equation}
We can now amplify this bound $T_2$ times while using the fact that the total fitness change is at most $\epsilon$ ($T_1$ and $T_2$ will be set after). 
\begin{align*}
T_2 \left(\frac{\alpha \Pr[A_{j;T}]}{4N} - 4/T_1\right) & \leq 
\sum_{\ell=0}^{T_2-1} \E\left[ \sum_{t=\ell*T_1}^{(\ell+1)*T_1} \tilde{f}(\mu^{t+1})-\tilde{f}(\nu^t)  \right] \frac{1}{1-n\epsilon} \\ 
& =  \E\left[ \sum_{t=0}^{T_1*T_2} \tilde{f}(\mu^{t+1})-\tilde{f}(\nu^t)\right] \frac{1}{1-n\epsilon}\\
& = \E\left[ \sum_{t=0}^{T_1*T_2} \tilde{f}(\mu^{t+1})-\tilde{f}(\nu^t) + 
\tilde{f}(\nu^{t})-\tilde{f}(\mu^{t})\right]\frac{1}{1-n\epsilon}\\
& = \E[ \tilde{f}(\mu^{T+1}) - \tilde{f}(\mu^0)]\frac{1}{1-n\epsilon}\\
& \leq \frac{\epsilon}{1-n\epsilon}
\end{align*}
Finally, we have 
$$ \Pr[A_{j;T}]  \leq \frac{4N\epsilon}{T_2(1-n\epsilon)\alpha} + \frac{8N}{\alpha T_1}.$$ 
Taking $T_1= 16N^2 n^4$ and $T_2 = \frac{\epsilon 8N^2 n^4}{1-n\epsilon}$ bounds the probability by $1/n^2.$
\end{proof}

Lemma \ref{lem:fast-vertex} tells us that with probability at least $1-2/n$, the population vector will be at a vertex after at most $T=O\left(\frac{\epsilon n^8 N^4}{1-n\epsilon}\right)$ steps, in which case $\tilde{f}(\mu^{T})\in \{1,1+\epsilon\}$.  On the other hand,       
Lemma \ref{lem:main} tells us that for any $T$, the probability that the total negative effect of the sampling exceeds $\alpha$ is at most $\beta$; since the fitness change for each selection step is non-negative (for our choice of $\epsilon$), we have that 
$$ \Pr[\tilde{f}(\mu^{T})\neq 1+\epsilon] = \Pr[ \tilde{f}(\mu^{T}) < \tilde{f}(\mu^{(0)}) - \alpha] \leq \beta$$
when $\tilde{f}(\mu^0) > 1+\alpha$. 

\end{document}